\documentclass[11pt]{article}%
\usepackage{eurosym}
\usepackage{bbm}
\usepackage{amsmath}
\usepackage{braket}
\usepackage{amsfonts}
\usepackage{amssymb}
\usepackage{graphicx}
\usepackage{pstricks}
\usepackage{pst-node}
\usepackage{pst-plot}%
\providecommand{\U}[1]{\protect\rule{.1in}{.1in}}
\providecommand{\U}[1]{\protect\rule{.1in}{.1in}}
\newcommand{\N}{{\mathbb N}}
\newcommand{\C}{{\mathbb C}}

\newcommand{\R}{{\mathbb R}}
\newcommand{\Z}{{\mathbb Z}}

\newcommand{\Pc}{{\mathcal P}}
\newcommand{\Hc}{{\mathcal H}}

\newcommand{\Lc}{{\mathcal L}}

\def\ov{\overline}
\newcommand{\ba}{\begin{eqnarray}}
\newcommand{\ea}{\end{eqnarray}}
\newcommand{\bas}{\begin{eqnarray*}}
\newcommand{\eas}{\end{eqnarray*}}
\newcommand{\be}{\begin{equation}}
\newcommand{\ee}{\end{equation}}

\newtheorem{theorem}{Theorem}
\newtheorem{proposition}[theorem]{Proposition}

\newtheorem{definition}[theorem]{Definition}
\newtheorem{conjecture}[theorem]{Conjecture}

\newenvironment{proof}[1][Proof]{\noindent\textbf{#1.} }{\ \rule{0.5em}{0.5em}}
\newtheorem{preremark}[theorem]{Remark}
\newenvironment{remark}{\begin{preremark}\rm}{\hfill$\Diamond$\end{preremark}}
\newtheorem{prenotation}[theorem]{Notation}

\numberwithin{equation}{section}
\numberwithin{theorem}{section}
\begin{document}

\title{{Quantization in singular real polarizations: K\"ahler regularization, Maslov correction and pairings }}
\author{Jo\~ao N. Esteves\thanks{Center for Mathematical Analysis, Geometry and Dynamical Systems, Instituto Superior T\'ecnico, Universidade de Lisboa, Av. Rovisco Pais, 1049-001 Lisboa; joao.n.esteves@tecnico.ulisboa.pt}, 
Jos\'e M. Mour\~ao\thanks{Lehrstuhl f\"ur Theoretische Physik III, FAU Erlangen-N\"urnberg and Center for Mathematical Analysis, 
Geometry and Dynamical Systems, Department of Mathematics, Instituto Superior T\'ecnico, Universidade de Lisboa, Av. Rovisco Pais, 1049-001 
Lisboa; jmourao@math.tecnico.ulisboa.pt} and Jo\~ao P. Nunes\thanks{Center for Mathematical Analysis, Geometry and Dynamical Systems, Department of Mathematics, Instituto Superior T\'ecnico, Universidade de Lisboa, Av. Rovisco Pais, 1049-001 Lisboa; jpnunes@math.tecnico.ulisboa.pt}}
\maketitle

\date

\begin{abstract}

We study the Maslov correction to semiclassical states by using a K\"ahler regularized
BKS pairing map from the energy representation to the Schr\"odinger representation. 
For general semiclassical states, the existence of this regularization is based on recently found families of K\"ahler polarizations degenerating 
to singular real polarizations and corresponding to special geodesic rays in the 
space of K\"ahler metrics.
In the case of the one-dimensional harmonic oscillator, we show that the  correct phases associated with caustic points of
the projection of the Lagrangian curves to the configuration space are correctly
reproduced. 

\end{abstract}
\tableofcontents

\section{Introduction}
\label{s1}

Though K\"ahler quantization is mathematically better defined, 
the quantization in (possibly singular) real or
mixed polarizations is frequently physically more interesting. This is partly due to the
fact that
the observables preserving mixed polarizations are likely to
be physically more interesting than those preserving a K\"ahler polarization.
In the present paper we continue along the lines proposed in 
\cite{BFMN, KW2, KMN1, Ki}, which motivate the definition 
of the
quantization for real or mixed polarizations  via degeneration
of  quantizations on suitable families of  K\"ahler polarizations. 

While for translation invariant real or mixed polarizations on a symplectic
vector space it is easy to construct families of K\"ahler polarizations
degenerating to the given one, for singular polarizations 
with lower dimensional (thus singular) fibers
it is usually not easy
to find such K\"ahler families explicitly (see the no-go Theorem \ref{tt34} and the
Conjecture \ref{cc35} below). 
Building up on previous works \cite{BFMN, KW2, KMN1, MN1}, 
in \cite{MN2} a general strategy is proposed
to find families of well behaved polarizations
degenerating to  a wide class of singular real polarizations corresponding to
the level sets of completely integrable systems. 

To obtain a direct link between standard half-form corrected quantization and the Maslov phases appearing 
in the definition of semiclassical states associated
with Lagrangian submanifolds, 
we adopt the following strategy. We consider cases for which the
given Lagrangian submanifold is a leaf of a 
(possibly singular) Lagrangian fibration 
 corresponding to
the level sets of a moment map $\mu$ of a completely integrable system. 
Denoting the corresponding real polarization by $\Pc_\mu$, 
the proposal of \cite{MN2} corresponds then to obtaining 
the Hilbert space ${\mathcal H}_\mu$, of quantum states in the polarization  $\Pc_\mu$,
as the infinite imaginary time limit $\sqrt{-1}s$ with $s \to +\infty$, of the family defined by applying the
imaginary time flow of the Hamiltonian vector field of the norm square of the moment map, $X_{||\mu||^2}$, 
to the Hilbert space corresponding to a starting K\"ahler
quantization. In order to relate the Schr\"odinger representation to this K\"ahler polarization, we consider the Thiemann
complexifier method \cite{Th1,Th2} adapted to geometric quantization in \cite{HK1, HK2, KMN2, KMN3, KMN4}. 
In the toric case, these families of polarizations  were first introduced in \cite{BFMN}. These families were also used 
in \cite{KMN1} to derive the Maslov shift of levels of the Bohr-Sommerfeld leaves.
The momentum space quantization $T^*K$, for compact  Lie groups $K$, was also defined in \cite{KW2} through
the infinite imaginary time limit of the flow of the Hamiltonian 
vector field of the norm squared of the (in general non-abelian)  moment 
map of the  action of $K$ on $T^*K$ (see also \cite{KMN2}).

The Maslov correction has been extensively studied (see, for example, \cite{EHHL, Go, GS, MF, Wo, Wu} and references therein).
In particular, in \cite{Wu} the holonomy of the natural projectively flat connection along geodesic triangles in the 
space of K\"ahler polarizations on $\R^{2n}$ invariant under translations, was shown to yield the triple Maslov index of
Kashiwara when the vertices of the triangle approach mutually transverse polarizations at geodesic infinity. 

In the present paper, we focus on obtaining  semiclassical states associated with real polarizations on $\R^{2n}$,  non--invariant
under translations. We are particularly interested in polarizations for which only K\"ahler regularizations of the 
second type exist (see section \ref{s333} below). We propose a general formalism in section \ref{s4.1} and apply
it to the harmonic oscillator in section \ref{s4.2}.
\section{Preliminaries}
\label{s2}

Let  $(M, \omega, I)$ be a connected  K\"ahler manifold such that $ [\frac{\omega}{2\pi\hbar}],  \frac 12 c_1(M)  \in H^2(M, \Z)$
so that the  canonical line bundle $K_I$ has $I$-holomorphic square roots. Let $\sqrt{K_I}$ denote one such square root
with Chern connection $\nabla^I$ and
let us fix a complex line bundle $L \rightarrow M$ with first Chern class
$c_1(L)= [\frac{\omega}{2\pi\hbar}]$. We consider on $L$ a connection $\nabla$ with curvature
$F_{\nabla}=-\frac{i}{{\hbar}} \omega$ and a compatible Hermitian structure $h^L$.
The half-form corrected quantum Hilbert space corresponding
to $I$  is then 
\be
\label{jsghs}
\mathcal{H}_{{\mathcal{P}}_{I}}= \overline{\left\{s\in\Gamma(L \otimes \sqrt{K_I}):  \left( \nabla_{{\mathcal{P}}_{I}} \otimes 1+ 
1 \otimes  \nabla^I_{{\mathcal{P}}_{I}} \right) 
\ s=0 \right\}},
\ee
where $\Pc_I$ denotes the polarization generated by $I$-anti-holomorphic vector fields and
the bar denotes completion with respect to the inner product  defined by
\begin{equation}\label{innerproduct}
 \langle \sigma,\sigma'\rangle = \frac{1}{(2\pi \hbar)^n}\int_X h^L(\sigma,\sigma') \frac{\omega^n}{n!}.
\end{equation}
In cases when the canonical bundle $K_I$ is trivial and $\Omega_I$ is a global trivializing
section we choose as   $\sqrt{K_I}$  the trivial square root and denote 
by $\sqrt{\Omega_I}$ one of the two trivializing sections of    $\sqrt{K_I}$
which square to $\Omega_I$.

The half-form corrected prequantization of a function $f \in C^\infty(M)$
is given by
\be
\nonumber
\hat f^{pQ} = i \hbar \nabla_{X_f} \otimes 1 + f \otimes 1 + i \hbar 1 \otimes  \Lc_{X_f} ,   
\ee
where $X_f$ denotes the Hamiltonian vector field corresponding to $f$,
or, if a local trivializing section $\sigma$ of $L$ is given, such that
\be
 \label{gauge}
\nabla \sigma = \frac i{\hbar}\Theta \, \sigma , 
\ee
where $d \Theta = - \omega$, 
we obtain, in the local trivialization of $L$ defined by $\sigma$,
\be
\label{pQ}
\hat f^{pQ}  =  i \hbar {X_f} \otimes 1 - L_f \otimes 1 + i \hbar \, 1 \otimes  \Lc_{X_f} , 
\ee
where $L_f = \Theta(X_f) - f$  is called  the Lagrangian of $f$.


\section{K\"ahler regularizations}
\label{s333}

In the present section we study
regularizations associated with the imaginary time flow 
of hamiltonians, $h \in C^{\infty}(M)$, which we call regulators. 
Depending on the mixed polarization $\Pc$ we will consider regulators
of two types.

\begin{definition}
\label{dd1}
$\Pc$-regulators of the first type  or Thiemann (partial) complexifiers
are regulators $h$ for which,  there exists a $T>0$ such that
the polarization
$\Pc^h_{it} = e^{it {\Lc_{X_h}}} (\Pc)$
exists and is K\"ahler for $t  \in (0, T)$.
\end{definition}
 
In interesting families of examples, we can then define also a 
sensible limit of the corresponding Hilbert spaces of polarized quantum states, 
\be
\label{ret1}
\Hc_\Pc := \lim_{t\to 0} \Hc^h_{\Pc_{it}}.
\ee 
In the examples in \cite{KMN1,KMN2,KMN4}, the space of $\Pc$-polarized quantum states was already known and 
the limit actually recovers the correct Hilbert space. Conjecturally, however, one could possibly start with a badly behaved (and hence difficult to quantize directly) polarization $\Pc$ and {\it define} $\Hc_\Pc$ as a limit of well-behaved quantizations in K\"ahler polarizations.

Regulators of the first type were introduced by Thiemann 
in the context of non-perturbative quantum gravity \cite{Th1, Th2} to transform
the $SU(2)$ spin connection to the $SL(2, \C)$ Ashtekar connection.
The prototypical example in finite dimensions  is that of  the vertical polarization
on a cotangent bundle $M=T^*X$ of a compact manifold $X$. Hall and Kirwin \cite{HK1, HK2} showed, both 
for the canonical symplectic form $\omega_c$ and for a symplectic form 
modified by a magnetic field, $\omega_b + B$,
that the
imaginary time flow of the kinetic energy, $h=E_\gamma$, corresponding  to a Riemannian metric $\gamma$
on $X$ defines, at $t=1$ and on a tubular neighborhood of the zero section, a
 K\"ahler structure, which, for $B=0$, coincides with
 the adapted K\"ahler structure introduced
by Guillemin-Stenzel \cite{GS1, GS2} and Lempert-Sz\"oke \cite{LS}. In the cases
when the K\"ahler structure extends to $T^*X$, as is the case of compact Lie groups $X$ with bi-invariant metric, $E_\gamma$ can
be used as a regulator of the first type \cite{KMN2}.

\begin{remark}
Regulators of the first type however do not allow to obtain 
K\"ahler regularizations of many real polarizations as we show 
below in Theorem \ref{tt34}. See also the Conjecture \ref{cc35}.
\end{remark}

\begin{definition}
\label{dd2}
We call $h  \in C^{\infty}(M)$ a $\Pc$--regulator of the second type or (partial) decomplexifier if there 
exist a  polarization $\Pc_0$  such that
the polarization $\Pc^h_{it} = e^{it {\Lc_{X_h}}} (\Pc_0)$
exists and is K\"ahler for $t  > 0$ and 
$$
\lim_{t \to +\infty}   \Pc^h_{it} = \Pc,
$$
in an open, dense subset of $M$. 
\end{definition}
Then, as above, one can look for a sensible definition of the space of $\Pc$-polarized quantum states 
by considering the limit 
\be
\label{ret2}
\Hc_\Pc := \lim_{t\to \infty} \Hc_{\Pc^h_{it}},
\ee 
in an appropriate sense.

The need for regulators of the second type comes 
from the difficulty in finding regulators of the first
type for example for real polarizations with compact fibers. 
In fact, we can prove easily the following result concerning the important case of completely integrable 
systems on compact manifolds. 

\begin{theorem}
\label{tt34}
Let $(M, \omega)$ be a compact real analytic completely integrable system defined by $n$ Hamiltonian functions $H_1,\dots,H_n$ in 
Poisson involution, with $dH_1 \wedge \cdots \wedge dH_n \neq 0$ on an open dense subset of $M$.
Let $\Pc$ be the real (necessarily singular) polarization with integral leaves corresponding to the level sets of $\mu = (H_1, \dots, H_n)$. 
Then there can be no real analytic $\Pc$--regulator of the first type.
\end{theorem}

\begin{proof}
Recall that $\Pc$ is pointwise generated by the global Hamiltonian vector fields $X_{H_j}, j=1,\dots, n$. Suppose that there
exists 
 a $\Pc$--regulator  of the first type,  $h \in C^\omega(M)$. From \cite{MN1}, it then follows that there exists $\epsilon >0$ such that
$ \Pc^h_\tau = e^{\tau \Lc_{X_h}} \Pc$ is, for all $\tau \in \C \ : |\tau| < \epsilon$,  a polarization generated by the (complex) 
Hamiltonian vector fields of the 
global functions $H_j^\tau = e^{\tau {X_h}} H_j$. Then there exists a 
$\epsilon' \leq \epsilon$ such that $ \Pc^h_{it}$ is K\"ahler for all $t \, : \, 0<|t| < \epsilon'$ and $H_j^{it}$ are nonconstant global
holomorphic functions which contradicts the compacteness of $M$. 
\end{proof}

\begin{conjecture}
\label{cc35}
We conjecture that do not exist regulators of the first type  for singular polarizations
$\Pc$ 
 such that
 there exist points $x\in M$ for which $\Pc_x$ is an isotropic non-Lagrangian subspace of $T_xM \otimes \C$.
\end{conjecture}

Fortunately, as shown in \cite{KMN4, MN2}, there are regulators of the second type for many  of the above 
examples. They are given by strongly convex functions of the Hamiltonians $H_j$.

\section{Schr\"odinger semiclassical states and Maslov phases}
\label{s3}

\subsection{K\"ahler regularized semiclassical states}
\label{s4.1}

Let $\Lc \subset T^* \R^n$ be a compact closed Lagrangian submanifold and consider the 
Schr\"odinger representation, that is the prequantum line bundle $ L = T^* \R^n \times \C$ with global
trivializing section the constant function, with connection $\nabla \, : \, \nabla 1 = \frac i\hbar  p dq =  \frac i\hbar \sum_{j=1}^n p_jdq_j $ and
\ba
\label{ee-srp}
\nonumber
\Pc_{Sch}&=&\langle\frac{\partial}{\partial p}\rangle_\C= \langle\frac{\partial}{\partial p_1}, \dots ,
\frac{\partial}{\partial p_n}\rangle_\C   \\
\Hc^Q_{Sch} &=& L^2(\R^n, d^n  q) \otimes \sqrt{d^n q } . 
\ea
 Our goal
in the present section will be to use K\"ahler regularization to construct a semiclassical
state $\psi_\Lc$  in the Schr\"odinger representation that is an approximate eigenvector 
of a quantum Hamiltonian $\hat h$ corresponding to the quantization of a function $h \in C^\infty(T^* \R^n)$
such that  
\be
\label{ee-lag}
\Lc \subset \left\{(q,p) \in \R^{2n}  \  : \ h(q, p) = E\right\}.
\ee   
$\psi_\Lc$ will therefore be an approximate solution of the eigenvalue equation
\be
\label{ese}
\hat h \psi = E \psi .
\ee 

Such states have been obtained  mainly by using the WKB method of constructing approximate
solutions of (\ref{ese}) and then imposing the Maslov correction to improve the solution (see, for example, 
\cite{Wo,GS,BW}). The Maslov correction changes the energy levels  (and therefore
the set of quantizable Lagrangians)  by correcting the
Bohr-Sommerfeld quantization conditions  and introduces phases in the caustic  points of the
projection $\pi(\Lc)$, where $\pi$ is the canonical projection $\pi : T^*\R^n  \rightarrow \R^n, \
\pi(q, p) = q$.  

To construct $\psi_\Lc$ with the help of K\"ahler regularization, we will 
consider K\"ahler regularizations of both the energy and Schr\"odinger 
representations, such that they are both deformed, through one-parameter families, 
to K\"ahler polarizations. We then use the limit of the BKS pairing map between the 
K\"ahler polarizations along these families, to define the pairing map $B$  from the energy 
representation to the Schr\"odinger representation.
We construct $\hat \psi_\Lc$ in the energy representation and
define $\psi_\Lc = B(\hat \psi_\Lc)$ (see (\ref{ee-bss}), (\ref{ee-rscs})).

For the vertical polarization, $\Pc_{Sch} = \langle \frac{\partial }{\partial p_j}, j=1, \dots, n \rangle_\C$, 
there are many Thiemann complexifiers or regulators of the first type. 
It follows from \cite{KMN2} that any strongly convex function of the
momenta is a $\Pc_{Sch}$ regulator of the first type. Functions of both
$p$ and $q$ can also be used as, for example,  the Hamiltonians of
harmonic oscillators which will be studied in \cite{Es}. Let $h_1$ denote
such a regulator and assume that, eventhough $h_1$  does not
preserve the vertical polarization (otherwise it could not be 
a $\Pc_{Sch}$--regulator of the first type), it has a natural
quantization on the Schr\"odinger representation,
$\hat h_1^{Sch}$. Then 
let 
\be
\label{ee-pit}
\Pc_{it}^{h_1} = e^{it \Lc_{X_{h_1}}} \, \Pc_{Sch}
\ee
and define the following  Kostant--Souriau-Heisenberg (KSH)
regularization map for Schr\"odinger states
\be
\label{ee-uit1}
U_1^{it} =   e^{\frac{t}{{\hbar}} \hat h_1^{pQ}}  \circ e^{-\frac{t}{{\hbar}} \hat h_1^{Sch}}   :
\, \Hc_{\Pc_{Sch}} \longrightarrow \Hc_{\Pc_{it}^{h_1}},
\ee
where $\hat h_1^{pQ}$ is given by (\ref{pQ}).
In the context of studying the equivalence of quantizations
for different polarizations,  
this map was introduced in \cite{KMN2, KMN3} for the case of cotangent bundles of a compact Lie group, 
while for toric symplectic manifolds it was considered in \cite{KMN4}, and, more generally, it is studied 
in \cite{MN2}. (For $K$ a compact Lie group, 
$M=T^* K$ and $h_1$ the Hamiltonian corresponding
to geodesic motion on $K$ for the bi-invariant metric, the
KSH map is equivalent to the Hall coherent state transform  \cite{Ha1, Ha2}, see \cite{KMN2}.)

Next, we will need to assume that $h$ is an Hamiltonian in a classically completely integrable system with 
integrals of motion defining a moment map, 
$\mu = (H_1, \dots , H_n):\R^{2n} \cong T^*\R^n \rightarrow \R^n$,
such that $\Lc$ is a corrected Bohr-Sommerfeld fiber. I.e., for some $c_0\in \R^n$,
\be
\label{ee-lag0}
\Lc = \Lc_{c_0} = \left\{(q, p) \in \R^{2n} \ : \ \mu(q,p) = c_0\right\}  .  
\ee
and
\be
\label{ee-ham}
h = F \circ \mu , 
\ee
for some $F \in C^\infty(\R^n)$ with $F(c_0)=E$.
Let $\Pc_\mu$ denote the real polarization having the level sets $\Lc_c, c \in \R^n$ as leaves. 

\begin{definition}
\label{dd-er}
We will call the  polarization $\Pc_\mu$ the energy polarization
and the corresponding quantization on $\Hc_{\Pc_\mu}$ the energy
quantization.
\end{definition}

A real polarization on $\R^{2n}$ having a compact fiber will typically have singular (lower dimensional) 
fibers and therefore, due to Conjecture \ref{cc35}, we do not
expect a K\"ahler regulator of the first type to exist for $\Pc_\mu$ so that 
we will need to consider a regulator of the second type for $\Pc_\mu$. This problem was studied
in the toric case in \cite{BFMN, KMN1}  and in general in \cite{MN2} and 
we will review now some of the results.

Let us assume that the level sets $\Lc_c$ are compact for noncritical values
$c\in \R^n$ and that a function $G  : \R^n \rightarrow \R$ exists
such that  $h_2 = G \circ \mu$ is strongly convex as a function of all 
action variables on equivariant neighborhoods of all regular fibers. 
This, plus some technical assumptions on the Fourier 
coefficients of local holomorphic functions for some initial
complex structure $J_0$,  imply 
that the polarization
\be
\label{ee-rer}
\Pc_{it}^{h_2} = e^{it \Lc_{X_{h_2}}} \ \Pc_0 , 
\ee
where $\Pc_0$ is the polarization associated with $J_0$, 
converges to $\Pc_ \mu$ as $t \to \infty$, 
$$
\lim_{t\to +\infty}  \Pc_{it}^{h_2}  = \Pc_\mu , 
$$
in an appropriate (weak) sense \cite{BFMN, MN2}.
We will also assume that there are 
 local $J_0$-holomorphic coordinates, $\{u_j\}_{j=1,\dots, n}$,
such that
pointwize, on a neighborhood of every point, one has
\be
\label{ee-coor}
\lim_{t \to \infty} \beta(t) \, du_1^{(it)} \wedge \dots \wedge du_n^{(it)} = dH_1 \wedge \dots \wedge dH_n \ ,
\ee
for some sooth, positive function $\beta \in C^\infty((0, \infty))$, where
$u_j^{(it)} =  e^{it {X_{h_2}}} (u_j)$.
Consider the  following modification of the KSH map, introduced in \cite{MN2}, 
\be
\label{ee-uit2}
U_2^{it} = e^{-\frac{t}{{\hbar}} \hat h_2^\mu} \circ  e^{\frac{t}{{\hbar}} \hat h_2^{pQ}},
\ee
where  $\hat h_2^{pQ}$ is defined in (\ref{pQ}) and $ \hat h_2^\mu$
is the following self adjoint operator 
\be
\label{ee-defh2}
\hat h_2^\mu \left( f \otimes \sqrt{du_1^{(it)} \wedge \dots \wedge du_n^{(it)}} \right)= G(\tilde H_1^{pQ}, \dots , \tilde H_n^{pQ}) \left( f \right)  \otimes   \sqrt{du_1^{(it)} \wedge \dots \wedge du_n^{(it)}}  \,  ,
\ee
densely defined on
$L^2(\R^{2n}) \otimes \sqrt{du_1^{(it)} \wedge \dots \wedge du_n^{(it)}}$,
where $\tilde H^{pQ}_j$  denotes the prequantization of $H_j$ 
without the half-form correction, $\tilde H^{pQ}_j = i \hbar X_{H_j} - L_{H_j}$ (compare with (\ref{pQ})).
The operator $U_2^{it}$ in (\ref{ee-uit2})
maps, for all $t > 0$, $\Hc_{\Pc_{0}}$ to, in general non-polarized,
subspaces of
$$L^2(\R^{2n}) \otimes \sqrt{du_1^{(it)} \wedge \dots \wedge du_n^{(it)}}.$$
 We will further assume that  the following limit, for every  $\psi \in \Hc_{\Pc_0}$,
\be
\label{ee-u2i}
U_2^{i \infty} (\psi) = \lim_{t \to \infty} U_2^{it} (\psi)
\ee
exists  and
 the resulting map, $U_2^{i \infty} \, : \, \Hc_{\Pc_0} \rightarrow \Hc_{\Pc_\mu}$, is an isomorphism onto the 
space of polarized Dirac delta distributions supported 
 on Maslov corrected Bohr-Sommerfeld  Lagrangian leaves (see \cite{BFMN, KMN1, MN2}).
Then the distributional section
\be
\label{ee-bss}
\hat \psi_{\Lc_{c_0}}(H,\theta) = \delta (H-c_0) \, e^{\frac{i}{{\hbar}} \tilde c_0 \cdot  \theta} \otimes \sqrt{d^n H}
 \in \Hc_{\Pc_\mu} , 
\ee
where $(H,\theta)$ are local action-angle coordinates, the phase factor $\tilde c_0 \in \hbar \Z^n$ 
corresponds to the uncorrected Bohr-Sommerfeld leaf, $H = \tilde c$, 
(compare with (\ref{ee-extr}) for the harmonic oscillator), is the image
of a uniquely defined section $ \tilde \psi_\Lc  \in \Hc_{\Pc_0}$,
\be
\label{ee-phl}
\hat \psi_{\Lc}=  U_2^{i \infty} \ (\tilde \psi_\Lc) \, .
\ee

Summarizing, our proposal to use K\"ahler regularization to construct
semiclassical solutions of (\ref{ese}) can be divided in the 
following steps.

\begin{itemize}

\item[1)]  \emph{Choice of  regulators $h_1, h_2$ and construction of the KSH maps in  $U_1^{it}$ in (\ref{ee-uit1}) and $U_2^{it}$ in (\ref{ee-uit2}):}

\begin{itemize}

\item[i)] Choose the Thiemann complexifier or  regulator of first type, $h_1$, for the Schr\"odinger polarization
and define the one-parameter family of K\"ahler polarizations (\ref{ee-pit}) 
and K\"ahler regularizations of the Schr\"odinger representation (\ref{ee-uit1}). A standard choice is 
the free particle Hamiltonian, $h_1(q, p)=\frac12 ||p||^2$, for which
$\Pc_{it}^{h_1} = \langle X_{z_1^{(it)}}, \dots , X_{z_n^{(it)}}\rangle_\C$, where
$z^{(it)} = q + it p.$ {Another interesting possibility, studied in \cite{Es},
consists in using the regulator of second type used for the
polarization $\Pc_\mu$ also as Thiemann complexifier,
i.e. $h_1=h_2=h$.} Find $U_1^{it}$ in (\ref{ee-uit1}).

\item[ii)]Choose a $\Pc_\mu$--regulator of second type, $h_2$, for the energy 
representation of
 Definition \ref{dd-er}, define 
 the one-parameter family of K\"ahler polarizations (\ref{ee-rer})
and the K\"ahler regularized Hilbert space
as the image,  $U_2^{it} (\Hc_{\Pc_0})$, of (\ref{ee-uit2}), for large $t$,
leading  to (\ref{ee-u2i}),   (\ref{ee-bss}) and  (\ref{ee-phl}).
As mentioned above, from \cite{BFMN, KMN1, KMN4, MN2}, it follows
that $h_2$ should be a function on $\R^{2n}$ that is strongly convex in the
action coordinates. 
\end{itemize}

\item[2)]  \emph{BKS pairing between 
the Schr\"odinger representation and the energy representation
of Definition \ref{dd-er}:} 

For the states
$\psi_1 \in \Hc_{Sch}$ and 
$\psi_2 \in \Hc_{\mu}$ define their time $(t_1,t_2)$ regularized
BKS pairing as
\be
\label{ee-rbks}
\langle  \psi_1,  \psi_2     \rangle^{t_1,t_2} =
\langle U_1^{it_1}  (\psi_1), U_2^{it_2}  (\tilde \psi_2)  \rangle_{BKS}   , 
\ee
where 
$\tilde \psi_2$ is the state in $\Hc_{\Pc_0}$ mapped
to $\psi_2$ by $U^{i \infty}_2$, that is
$\tilde \psi_2 = (U^{i \infty}_2)^{-1}(\psi_2)$ and
the pairing in the right hand side is the usual half-form corrected
 BKS pairing of geometric
quantization. The BKS pairing between states in $\Hc_{Sch}$ and $\Hc_{\Pc_\mu}$  
is then defined by the following limit of K\"ahler regularized pairings
\be
\label{ee-rbks2}
\langle  \psi_1,  \psi_2     \rangle_{BKS} = \lim_{(t_1, t_2)  \to (0,  \infty) } \,\langle  \psi_1,  \psi_2     \rangle^{t_1,t_2}  \, , 
\ee
in case the limit exists $\forall \psi_1 \in  \Hc_{Sch}, \psi_2 \in
 \Hc_{\Pc_\mu}$. 

\item[3)]  \emph{Definition of the semiclassical state $\psi_\Lc$:}

If the limit in (\ref{ee-rbks2}) exists and is continuous on
$\Hc_{Sch} \times  \Hc_{\Pc_\mu}$,
 it defines a pairing map 
\be
\label{ee-pm}
B:\Hc_{\Pc_\mu}  \rightarrow  \Hc_{Sch} . 
\ee
The semiclassical state corresponding to $\Lc$ via the  K\"ahler regularization of the
pairing is then defined to be
\be
\label{ee-rscs}
\psi_\Lc = B(\hat \psi_\Lc) , 
\ee
where $\hat \psi_\Lc$ is the state defined in (\ref{ee-bss}). So, $\psi_\Lc$ is the unique state in $\Hc_{Sch}$ such that 
$$
\langle \psi,\psi_\Lc\rangle_{\Hc_{Sch}} = \langle \psi, \hat\psi_\Lc\rangle_{BKS}, \forall \psi \in \Hc_{Sch}.
$$
\end{itemize}

\begin{remark}
The pairing map $B$ and thus the semiclassical states, could in principle 
depend on the K\"ahler regularizations, though we would expect to obtain the 
same result for generic choices. 
\end{remark}

\subsection{Maslov phases for the harmonic oscillator}
\label{s4.2}

Let us illustrate the general method of the previous section in the 
case of the one-dimensional harmonic oscillator. We have $n=1,  M= \R^2, \omega = dq \wedge dp, 
\Theta = pdq$,  $L = \R^2 \times \C$, $\nabla \, 1 = \frac i\hbar \Theta$ and $h(q,p)= \mu = H =  \frac 12 (p^2+ q^2)$. As in (\ref{ee-srp}), we have
\be
\nonumber
\Hc_{Sch}= L^2(\R, dq) \otimes \sqrt{dq} .
\ee
For $\Hc_{\Pc_\mu}$, since $h$ is the action variable for the standard toric structure in $\R^2$, 
we know from \cite{KMN1} that
\be
\label{ee-ext2}
\Hc_{\Pc_\mu} = \langle \delta(h - \hbar (m+ \frac 12) ) \, e^{-\frac i{2\hbar} \,  pq} \, e^{i m \theta}, \, m \in \N_0 \rangle_\C \otimes \, \sqrt{dh} \, .
\ee

Let us now follow the three steps described in the previous section, in this example.

\begin{itemize}
\item[1)] \emph{Choice of  regulators $h_1, h_2$ and construction of $U_1^{it}$ in (\ref{ee-uit1}) and $U_2^{it}$ in (\ref{ee-uit2}):}

\begin{itemize}
\item[i)] 
Let us choose the standard Thiemann complexifier for the
Schr\"odinger representation, $h_1(q, p) = \frac 12 p^2$. Then
\bas
\nonumber
\label{ee-pith}
\Pc_{it}^{h_1} &=&  e^{it \Lc_{X_{h_1}}} \Pc_{Sch} =  e^{it \Lc_{X_{h_1}}} \langle X_q\rangle_\C = 
\langle X_{e^{it {X_{h_1}}}(q)}\rangle_\C = \\
&=&  \langle X_{q +itp}\rangle_\C  =   \langle X_{z^{(it)}}\rangle_\C 
\eas
 where
$z^{(it)} = q + it p,$ defines a K\"ahler polarization for all $t>0$, confirming that $h_1$ is
a $\Pc_{Sch}$--regulator of the first type. This is the simplest example of imaginary time
geodesic flows starting at the Schr\"odinger polarization and  leading (at time $= \sqrt{-1}$) to adapted K\"ahler structures on tubular  neighbourhoods (the whole $T^* \R$ in the present case) of the zero section of cotangent bundles of
compact Riemannian manifolds \cite{HK1}.
The equation for $\Pc_{it}^{h_1}$--polarized sections of $ L$ reads,
$$
\nabla_{X_{z^{(it)}}} \,  \psi  = 0 \Leftrightarrow    \left(- \frac {\partial}{\partial p} + it \frac {\partial}{\partial q} - p \frac t\hbar \right) \psi (q, p) = \left(2it \frac {\partial}{\partial \bar z^{it}} -p\frac{t}{\hbar}\right)\psi(q,p)=0, 
$$
and therefore
\be
\label{ee-hit1}
\Hc_{\Pc_{it}^{h_1}} = \left\{ f(z^{(it)})  \, e^{- \frac t{2\hbar} p^2} \,   \otimes \sqrt{dz^{(it)}} \, :  \, 
\int_{\R^2} \, | f(z^{(it)})|^2 \, e^{- \frac t\hbar p^2} \, dqdp   < \infty  \right\} \, ,
\ee
where $f$ is $\Pc_{it}^{h_1}$-holomorphic.

Let us now obtain the corresponding K\"ahler regularization maps, $U_1^{it}$.
{}From (\ref{pQ}) we obtain
\be
\label{ee-h1p} 
 \hat p^{pQ} = i \hbar \left(\frac{\partial}{\partial q} \otimes 1 
+  1 \otimes \Lc_{ \frac{\partial}{ \partial q} }\right) \, \hbox{and} \quad 
 \hat h_1^{pQ} = i \hbar  \, \left(p \frac{\partial}{\partial q} \otimes 1
+  1 \otimes \Lc_{p \, \frac{\partial}{ \partial q} }\right) -
\frac {p^2}2 \otimes 1 .
\ee
We see that, due to the fact that $p$ preserves the
Schr\"odinger polarization, $\hat p^{pQ}$ acts on the
Schr\"odinger representation. We can then define
\be
\label{ee-h1s}
\hat h_1^{Sch} = \frac{(\hat p^{pQ})^2}2|_{ \Hc^Q_{Sch}} = - \frac{\hbar^2}2 \, \frac{\partial^2}{\partial q^2} \otimes 1.
\ee
It is convenient to act with $U_1^{it}$ on $\psi \in \Hc_{Sch}$, written in the 
form
$$
\psi(q) \otimes \sqrt{dq} = \frac 1{\sqrt{2\pi}} \int_\R \, e^{ \frac i\hbar \, p_0q} \, \tilde \psi(p_0) \, dp_0 \, \otimes \sqrt{dq} .
$$
{}From (\ref{ee-uit1}),  (\ref{ee-h1p}) and  (\ref{ee-h1s}) we obtain for
$U_1^{it} (\psi) \in \Hc_{\Pc_{it}^{h_1}}$,
\be\label{t1state}
U_1^{it} (\psi)(q,p) =  \frac 1{\sqrt{2\pi}} \int_\R  \, e^{ \frac i\hbar \, p_0q} \, e^{- \frac{t}{2{\hbar}} (p{+}p_0)^2} \,  \tilde \psi(p_0) \, dp_0 \, \otimes \sqrt{dz^{(it)}}   \, .
\ee

\item[ii)]  For the choice of $h_2$ notice that $h=H=\mu$  is an action coordinate in $\R^2 \setminus \{0\}$ 
generating a global $S^1$ action.
The
angle coordinate is the polar angle $\theta = \arctan(p/q)$. 
We can therefore choose as $h_2$ a strongly convex function of
$H=h$, e.g. $$h_2(q, p) = \frac 12 H^2 = \frac 18 (p^2+q^2)^2.$$ 
Let $w=z^{(i)}=q+ip = \sqrt{2h}\, e^{i \theta}$ and
choose as starting polarization, the $S^1$--invariant 
toric polarization $\Pc_0 = \langle X_w \rangle_\C= \langle\frac{\partial}{\partial \ov w}\rangle_\C$. 
We have $X_{h_2} = - h \frac{\partial}{\partial \theta}$ and therefore the one-parameter
family of polarizations obtained by flowing with this vector field in imaginary time 
is also toric and a simple particular example of those studied in \cite{BFMN}
and \cite{KMN1}
\bas
\nonumber
\label{ee-pith2}
\Pc_{it}^{h_2} &=&  e^{it \Lc_{X_{h_2}}} \Pc_{0} =  e^{it \Lc_{X_{h_2}}} \langle X_w\rangle_\C = 
\langle X_{e^{-it h \frac{\partial}{\partial \theta}}w}    \rangle_\C = \\
&=&    \langle X_{w^{(it)}}    \rangle_\C \, ,
\eas
 where
\begin{equation}\label{www}
\frac { w^{(it)}}{\sqrt{2}} \, =    \sqrt{h} \,  e^{t h}e^{i \theta} =  e^{\frac 12 \log(h) + th + i \theta}
=  e^{\frac {dg}{dh}+ i \theta},
\end{equation}
and $g(h) = \frac 12 h\log(h) - \frac h2 + t \frac {h^2}{2}$ is the strongly convex toric symplectic potential.
(See e.g. \cite{BFMN}).
A local  $J_0$--holomorphic coordinate satisfying (\ref{ee-coor}) is
$u =  \log(w/\sqrt{2}) = t h +  \frac 12 \log(h) + i \theta$ with $\beta(t) = 1/t$. 
To find the Hilbert space it will be convenient to use the $S^1$-invariant trivializing
section (we will return to the trivialization defined  by $\sigma(q,p)=1$ when we calculate the pairing of states of different polarizations)
$\tilde \sigma = e^{- \frac{i}{2\hbar}qp} \, \sigma$ so that 
$$
\nabla \, \tilde \sigma = - \frac i\hbar \, h d \theta \, \tilde \sigma \, . $$
The equation for $\Pc_{it}^{h_2}$-polarized sections then reads
\be
\label{ee-pola}
\nabla_{\frac{\partial}{\partial \bar u}} \, \psi = 0 \Leftrightarrow  \left(\frac {\partial}{\partial v} +
 i \frac {\partial}{\partial \theta}  + \frac h\hbar \right) \tilde \psi(q, p) = 0 .
\ee
Since $g$ is strictly convex and $v= \frac{dg}{dh}$, the inverse Legendre map is
$h = \frac{dk}{dv}$, where $k(v) = h(v)v - g(h(v))$ is the K\"ahler potential. 
We see that the solutions of (\ref{ee-pola}) are given by
\be
\nonumber
\tilde \psi (q, p) = f(w^{(it)}) \, e^{- \frac k\hbar}  =  f(w^{(it)}) \, e^{- \frac 1{2\hbar} (th^2+h)}, 
\ee
where $f$ is an arbitrary $\Pc_{it}^{h_2}$-holomorphic function. 
Therefore, we have, in the trivialization defined by $\tilde \sigma$,
\be
\label{ee-hit2}
\tilde \Hc_{\Pc_{it}^{h_2}} = \left\{ f(w^{(it)})  \, e^{- \frac 1{2\hbar} (th^2+h)} \,   \otimes \sqrt{dw^{(it)}} \, :  \, 
\int_{\R^+ \times S^1} \, | f(w^{(it)})|^2 \, e^{- \frac 1\hbar  (th^2+h)} \, \sqrt{g''(h)} \, dhd\theta   < \infty  \right\} \, ,
\ee

Let us now obtain $U_2^{(it)}$  in   (\ref{ee-uit2}). From (\ref{pQ}) we obtain
\be
\label{ee-h2p}
\hat h^{pQ} = - i \hbar  \left( \frac{\partial}{\partial \theta} \otimes 1 + \,  1 \otimes \Lc_{\frac{\partial}{\partial \theta}}\right) \, 
\hbox{and} \quad 
\hat h_2^{pQ} = - i \hbar \left(  h \, \frac{\partial}{\partial \theta} \otimes 1 +   1 \otimes \Lc_{h \frac{\partial}{\partial \theta}}\right) - \frac {h^2}2 \,  \otimes 1  . 
\ee
The sections  in the monomial basis of $\tilde \Hc_{\Pc_{it}^{h_2}}$,
\ba \label{ee-phit}
\varphi_m^{(it)} &=& {a_m} (w^{(it)})^m \,  \, e^{- \frac 1{2\hbar} (th^2+h)} \,   
\otimes \sqrt{dw^{(it)}} = \\
\nonumber  &=& {a_m 2^{\frac{m}{2}+\frac14}}h^{\frac{m}2+\frac14}\, e^{- \frac h{2\hbar}} \, e^{- \frac t{2\hbar} ( (h- \hbar (m + \frac 12))^2)} e^{\frac{t\hbar}{2}(m+ \frac 12)^2} \, e^{i(m+1/2) \theta} \, \sqrt{du^{(it)}} ,  \quad  m \in \N_0, 
\ea
where ${a_m = 2^{-\frac{m}{2}-\frac14} (2\pi\hbar)^{-\frac12} (\hbar (m+\frac12))^{-\frac{m}{2}-\frac14}e^{\frac{m}{2}+\frac12}}$, 
form an orthogonal basis of  eigensections of $\hat h^{pQ}$ and of $\hat h_2^\mu$ (see (\ref{ee-defh2})), with  
\be
 \label{ee-h22q}
\hat h_2^{\mu}(\varphi_m^{(it)}) = \frac 12 \left(\hbar (m + \frac 12)\right)^2 \, \varphi_m^{(it)} .
\ee
The constants $a_m$ in (\ref{ee-phit}) are chosen to have $ \varphi_m^{(0)}=\tilde \psi_{\Lc_m}$ (see (\ref{ee-phl})).
In this case $\hat h_2^{\mu}$ acts on the space of polarized sections 
because $h$ preserves the polarizations $\Pc^{h_2}_{it}$ for every $t \in [0, \infty)$.
Eventhough $h_2$ itself does not preserve the polarization the operator $\hat h_2^\mu$ is defined
through $\tilde h$ in (\ref{ee-defh2}) and therefore has a well defined action 
on the Hilbert spaces $\Hc_{\Pc^{h_2}_{it}}$.  
Note that (\ref{ee-phit}) is a local expression of a global $\Pc^{h_2}_{it}$--holomorphic  section of the half-form corrected prequantum bundle in spite of the 
factor of $e^{\frac{i}{2}\theta}$. In fact, as explained in Section 3 and in the Appendix of \cite{KMN1}, this factor gets 
canceled against a similar factor arising from the fact that $du^{(it)}$ is a meromorphic section of the canonical bundle having 
a pole of order 1 at the origin. 
(Or, equivalently, (\ref{ee-phit}) is written in a frame of the half-form corrected prequantum bundle 
which has a square root ramification divisor at the origin.)

To the Lagrangian cycles
$$
\Lc_m = \{(q,p)\in \R^2: q^2+p^2 = \hbar (2m+1)\}, \, m\in \N_0,
$$
 as in (\ref{ee-phl}), there will then correspond the state $\tilde \psi_{\Lc_m}= \varphi_m^{(0)}$, as shown below.

{}From (\ref{ee-h2p}) and (\ref{ee-phit}) we also obtain that $e^{\frac t\hbar \, \hat h_2^{pQ}} (\varphi_m^{(0)}) = \varphi_m^{(it)}$  and therefore
\be\label{t2monomialstate}
U_2^{it} (\varphi_m^{(0)}) =  
 e^{-\frac t\hbar \, \hat h_2^{\mu}} \circ e^{\frac t\hbar \, \hat h_2^{pQ}}   (\varphi_m^{(0)}) = 
$$
$$
= {a_m 2^{\frac{m}{2}+\frac14}}h^{\frac{m}2+\frac14} 
e^{- \frac h{2\hbar}} \, e^{- \frac t{2\hbar} ( (h- \hbar (m + \frac 12))^2) } \, e^{i(m+1/2) \theta} \, \sqrt{\left( 1/{2h}+t\right)dh + i d\theta} . 
\ee
\end{itemize}

When taking the limit $t\to +\infty$, following \cite{KMN1} and taking care of the fact that $\sqrt{du^{(it)}}$ effectively carries 
a factor of $e^{-\frac{i}{2}\theta}$ as remarked above, we obtain
\be
\label{ee-extr}
\hat \psi_{\Lc_m}= \lim_{t \to \infty} U^{it}_2 (\varphi_m^{(0)}) = \delta(h-\hbar (m+\frac12)) \, e^{im\theta} \, \sqrt{dh},\,\, m\in \N_0  
\ee
so that indeed, $\tilde  \psi_{\Lc_m} = \varphi_m^{(0)}$.

\item[2)] \emph{BKS pairing between the Schr\"odinger representation and the energy representation:} 

\begin{proposition}
For the holomorphic forms in (\ref{t1state}) with $t=t_1$ and (\ref{ee-phit}) with $t=t_2$, we obtain
\begin{equation}\label{formpair}
\left(\frac{i}{2}\right) dz^{(it_1)}\wedge \left(\left(\frac{1}{2h}+t_2\right)dh-id\theta\right) = \left(\frac{i}{2}\right)
\left(\frac{(p-iq)(1+t_1)}{p^2+q^2}+t_2(p-it_1q)\right) \omega.
\end{equation}
\end{proposition}

\begin{proof}The result follows directly from $z^{(it_1)}= q+it_1 p$ and (\ref{www}) with $t=t_2$.
\end{proof}

\begin{proposition}The pairing in (\ref{ee-rbks2}) for the harmonic oscillator is given by
\begin{eqnarray}\label{harmonicpairing}\nonumber
&\langle \psi, \hat \psi_{\Lc_m}\rangle_{BKS} =
\lim_{(t_1,t_2)\to (0,+\infty)}\langle U_1^{it} (\psi)(q,p) , e^{ipq/2\hbar}U_2^{(it)} (\varphi_m^{(0)})\rangle_{BKS} =
\\
&=\sqrt{\frac{i}{2}} \int_{\R^2}  \psi(q) e^{-ipq/2\hbar}e^{-im \theta} \sqrt{p} \, 
\delta\left(h-\hbar(m+\frac12)\right)dqdp,
\end{eqnarray}
where $\psi\in \Hc_{Sch}$.
\end{proposition}

\begin{proof}
From (\ref{t1state}), (\ref{t2monomialstate}) and (\ref{formpair}) we obtain the BKS pairing
$$
\langle U_1^{it} (\psi)(q,p) , e^{ipq/2\hbar}U_2^{(it)} (\varphi_m^{(0)})\rangle_{BKS} =
$$
$$
= \frac{1}{\sqrt{2\pi}}\sqrt{\frac{i}{2}} a_m 2^{\frac{m}{2}+\frac14}\int_{\R^2} \int_\R  
e^{ \frac i\hbar \, p_0q} e^{- \frac{t_1}{2{\hbar}} (p{+}p_0)^2} \tilde \psi(p_0) h^{\frac{m}2+\frac14} e^{-ipq/2\hbar}
e^{- \frac h{2\hbar}} e^{- \frac{t_2}{2\hbar} ( (h- \hbar (m + \frac 12))^2) } e^{-im \theta} \cdot 
$$
$$
\cdot \left(\frac{(p-iq)(1+t_1)}{p^2+q^2}+t_2(p-it_1q)\right)^\frac12  dp_0 dqdp.
$$

Due to the gaussians, the integrals are convergent and bounded and, therefore,
the limit $(t_1,t_2)\to (0,+\infty)$ exists and can be taken inside the integral.

Taking the limit $t_2\to +\infty$ gives,
$$
\lim_{t_2\to +\infty} \langle U_1^{it} (\psi)(q,p) , e^{ipq/2\hbar}U_2^{(it)} (\varphi_m^{(0)})\rangle_{BKS} =
\sqrt{\hbar} a_m 2^{\frac{m}{2}+\frac14}\sqrt{\frac{i}{2}} \left(\hbar (m+\frac12)\right)^{\frac{m}{2}+\frac14} e^{-(\frac{m}{2}+\frac14)}\cdot
$$
$$
\cdot \int_{\R^2} \int_\R  
e^{ \frac i\hbar \, p_0q} e^{- \frac{t_1}{2{\hbar}} (p{+}p_0)^2} \tilde \psi(p_0) e^{-ipq/2\hbar}e^{-im \theta} 
(p-it_1q)^{\frac12} \delta\left(h-\hbar(m+\frac12)\right)dp_0 dqdp.
$$
The limit $t_1\to 0$ and the integration in $dp_0$ then produces the result.
\end{proof}

\item[3)] \emph{The semiclassical state $\psi_\Lc$:}

From above we obtain the following
\begin{proposition}
Let $\Lc_m, m\in \N_0$, be the Lagrangian cycles where $h=\hbar (m+\frac12)$, as above. 
The pairing map in (\ref{ee-pm}) and (\ref{ee-rscs}) for the harmonic oscillator reads
\begin{equation}
B:\Hc_{\Pc_\mu}\to \Hc_{Sch}
\end{equation}
\begin{equation} \nonumber
\psi_{\Lc_m}(q) = B (\hat \psi_{\Lc_m})(q) = \psi^+_{\Lc_m}(q) + \psi^-_{\Lc_m}(q),
\end{equation}
such that $\psi^+_{\Lc_m}, \psi^-_{\Lc_m}$ have support in $[-\hbar(2m+1),\hbar(2m+1)]$ where they are given by 
\begin{equation} \nonumber
\psi^+_{\Lc_m}(q) = \sqrt{\frac{i}{2}} (\hbar(2m+1)-q^2)^{-\frac14} e^{\frac{-i}{2\hbar}q\sqrt{\hbar (2m+1)-q^2}}e^{im\arctan {\frac{\sqrt{\hbar (2m+1)-q^2}}{q}}}\otimes \sqrt{dq}
\end{equation}
and
\begin{equation}\label{wkbwave}
\psi^-_{\Lc_m}(q) = {\sqrt{\frac{i}{2}}}e^{i\frac{\pi}{2}}(\hbar(2m+1)-q^2)^{-\frac14}e^{\frac{i}{2\hbar}q\sqrt{\hbar (2m+1)-q^2}}e^{-im\arctan {\frac{\sqrt{\hbar (2m+1)-q^2}}{q}}}\otimes \sqrt{dq}.
\end{equation}
\end{proposition} 

\begin{proof}Let $\psi$ in (\ref{harmonicpairing}) be a continuous function in $L^2(\R)$. Then, since the inner product 
in $\Hc_{Sch}$ is given by integration in $q$, we obtain from (\ref{ee-pm}), (\ref{ee-rscs}) and (\ref{harmonicpairing}), 
\begin{equation}
\psi_{\Lc_m}(q) = B (\hat \psi_{\Lc_m})(q) = \psi^+_{\Lc_m}(q) + \psi^-_{\Lc_m}(q),
\end{equation}
where
\begin{equation}
\psi^+_{\Lc_m}(q) = {\sqrt{\frac{i}{2}}}\int_0^{+\sqrt{2\hbar(m+\frac12)}} e^{-i\frac{pq}{2\hbar}}e^{im\arctan \frac{p}{q}} \sqrt{p}\,\delta\left(\frac{p^2+q^2}{2}-\hbar (m+\frac12)\right) dp \otimes \sqrt{dq}
\end{equation}
and
\begin{equation}
\psi^-_{\Lc_m}(q) ={\sqrt{\frac{i}{2}}}e^{i\frac{\pi}{2}} \int_{-\sqrt{2\hbar (m+\frac12)}}^0 e^{-i\frac{pq}{2\hbar}}e^{im\arctan \frac{p}{q}} \sqrt{{-}p}\,\delta\left(\frac{p^2+q^2}{2}-\hbar (m+\frac12)\right) dp \otimes \sqrt{dq}
\end{equation}
and the result follows.
\end{proof}

We observe that $\psi_{\Lc_m}$, which is supported in the ``classically allowed region'' $[-\hbar(2m+1),\hbar(2m+1)]$, contains two contributions, weighted with a relative (Maslov) phase $e^{i\frac{\pi}{2}}$ which arises at the caustic points $q=\pm \sqrt{2\hbar (m+\frac12)}, p=0$, for the projection of $\Lc_m$ onto the $q$-axis. Moreover, by explicitly evaluating 
$\int pdq$ with $(q,p)\in \Lc_m$, we see that $\psi_{\Lc_m}$ (\ref{wkbwave}) has the form of the usual WKB wave function  in the classically allowed region. (See, for example, Chapter 7 of \cite{Me}.)

\end{itemize}

\large{\bf{Acknowledgements:}} The authors were partially supported by FCT/Portugal through the projects PEst-OE/EEI/LA009/2013, EXCL/MAT-GEO/0222/2012, \- PTDC/MAT/\-119689/2010, PTDC/MAT/1177762/2010. The first author was supported by the FCT fellowship
SFRH/BPD/77123/2011 and by a postdoctoral fellowship of the project
 PTDC/MAT/120411/2010. The second author thanks S. Wu for useful discussions and is thankful for generous support from the Emerging Field Project on Quantum Geometry from Erlangen--N\"urnberg University.

\providecommand{\bysame}{\leavevmode\hbox to3em{\hrulefill}\thinspace}

\end{document}